\documentclass[letterpaper, 10 pt, conference]{ieeeconf}  
\IEEEoverridecommandlockouts                             
\usepackage{siunitx}
\usepackage[framed,numbered,autolinebreaks,useliterate]{mcode}
\usepackage{amsmath,amssymb} 
    
\usepackage[utf8]{inputenc} 
\usepackage{hyperref} 
\usepackage{booktabs} 
\usepackage{graphicx} 
\usepackage{epstopdf}

\usepackage{wrapfig}
\usepackage[english]{babel} 
\usepackage{tabularx} 
\usepackage{fancyhdr} 
\usepackage{pdfpages} 
\usepackage{wrapfig} 
\usepackage{sidecap}
\usepackage{hyperref}
\usepackage{csquotes}
\usepackage[utf8]{inputenc}
\usepackage[english]{babel}
\usepackage{silence} 
\usepackage{textgreek} 
\usepackage{longtable} 
\usepackage{hyperref}
\usepackage{array}
\usepackage{tabulary}
\usepackage{float}
\usepackage[normalem]{ulem}
\usepackage{mathtools} 
\usepackage{url}
\usepackage{listings}
\usepackage{soul}
\usepackage{parskip}
\usepackage{listings}
\usepackage{color}
\usepackage{lscape}
\usepackage{bm}
\usepackage{tikz}
\usepackage{xcolor} 
\usepackage{lipsum}
\usepackage{hyperref}

\usepackage{algorithm}
\usepackage{algorithmic}

\usepackage{subfigure}
\usepackage[subfigure]{tocloft}

\usepackage{cleveref}
    \newtheorem{definition}{Definition}
    \crefname{definition}{definition}{definitions}

  \newtheorem{proposition}{Proposition}
    \crefname{proposition}{proposition}{propositions}
  \newtheorem{corollary}{Corollary}
    \crefname{corollary}{corollary}{corollary}
\newtheorem{remark}{Remark}
    
    \crefname{lemma}{Lemma}{Lemmas}

\definecolor{blue}{rgb}{0, 0, 1}
\definecolor{matlabblue}{rgb}{0, 0.447, 0.741}
\definecolor{red}{rgb}{1, 0, 0}
\definecolor{matlabgray}{rgb}{0.5, 0.5, 0.5}
\definecolor{matlablg}{rgb}{0.75, 0.75, 0.75}
\definecolor{matlabyellow}{rgb}{0.9290, 0.6940, 0.1250}
\definecolor{matlab_br_red}{rgb}{1, 0, 0}
\definecolor{matlab_br_blue}{rgb}{0, 0, 1}
\definecolor{matlabblack}{rgb}{0,0,0}
\definecolor{matlaborange}{rgb}{0.85, 0.325, 0.098}
\definecolor{matlabpurple}{rgb}{0.494, 0.1840 0.556}

\begin{document}
\title{\LARGE{\textbf{Toward Federated DeePC: borrowing data from similar systems}}}

\author{Gert Vankan, Simone Formentin and Valentina Breschi
\thanks{This work is partially supported by the FAIR project (NextGenerationEU, PNRR-PE-AI, M4C2, Investment 1.3), the 4DDS project (Italian Ministry of Enterprises and Made in Italy, grant F/310097/01-04/X56), and the PRIN PNRR project P2022NB77E (NextGenerationEU, CUP: D53D23016100001). It is also partly supported by the ENFIELD project (Horizon Europe, grant 101120657).}
	\thanks{Gert Vankan and Valentina Breschi are with Control Systems Group, Eindhoven University of Technology, 5612AZ Eindhoven, The Netherlands, (e-mails: \texttt{g.vankan@student.tue.nl}, \texttt{v.breschi@tue.nl}). Simone Formentin is with the Dipartimento di Elettronica, Informazione e Bioingegneria, Politecnico di Milano, 20133 Milan, Italy (e-mail: \texttt{simone.formentin@polimi.it}) } }

\maketitle

\begin{abstract}
     Data-driven predictive control approaches, in general, and Data-enabled Predictive Control (DeePC), in particular, exploit matrices of raw input/output trajectories for control design. These data are typically gathered only from the system to be controlled. Nonetheless, the increasing connectivity and inherent similarity of (mass-produced) systems have the potential to generate a considerable amount of information that can be exploited to undertake a control task. In light of this, we propose a preliminary \textit{federated} extension of DeePC that leverages a combination of input/output trajectories from multiple similar systems for predictive control. Supported by a suite of numerical examples, our analysis unveils the potential benefits of exploiting information from similar systems and its possible downsides.
\end{abstract}


\section{Introduction}\label{sec:intro}
Direct data-driven control methods have experienced a surge in popularity over the last years as they allow for bypassing the (arguably) most expensive and time-consuming task in advanced control design, i.e., modeling \cite{ogunnaike1996} (see, e.g., \cite{Dorfler2023a,Dorfler2023b} for discussions on recent advances on data-driven control, and \cite{Formentin2013} for a comparison between model-based and data-driven approaches). It is thus clear that the quality of the controller devised with these techniques heavily hinges on the data gathered to design it. Hence, given the central role played by data in data-driven control approaches and the costs associated with experimenting on modern systems, understanding how to efficiently and effectively use the data at disposal for control design becomes a necessity. 

Among existing data-driven control approaches, data-enabled predictive control (DeePC) \cite{coulsonDataEnabledPredictiveControl2019} is the first one using this data-driven control rationale for predictive control,  promising a relatively simple and fast method to design a predictive controller from raw data. Nonetheless, this method, as well as the subsequent ones proposed in the literature (see, e.g., \cite{coulsonDataEnabledPredictiveControl2019,berberich2020data,Chiuso2025} and the review in~\cite{verheijen2023handbook}, just to mention a few), rely on a set of batch data collected only from the system to be controlled. Instead, in real-world applications, one could often access data from similar systems, yet not identical, to the one under control (e.g., when considering mass-produced devices). These data can become an asset for control design, helping improve the controller's performance. Thus, it is essential to understand how to leverage them effectively. 

Understanding how to use data from multiple sources to tackle control problems is a recent yet not a new topic within control research (see the survey in~\cite{Weber2024}), with strategies to exploit data from multiple systems proposed for system identification (see, e.g., ~\cite{Smeraldo2023,Azzollini2023,Breschi2020}) and (model-based) control design~\cite{breschi2020collaborative,ferrarotti2021CloudSharing,Zeng2022}.  Nonetheless, integrating data from multiple similar systems into the context of data-driven (or data-enabled) predictive control remains unexplored.
\paragraph*{Contribution} By borrowing the term ``Federated" from machine learning \cite{konecnyFederatedOptimizationDistributed2015}, our goal in this work is to begin analyzing how to exploit data coming from multiple (similar) systems to design a predictive controller for a target (or nominal) system and whether their use can be advantageous. To this end, we propose using the raw data from multiple sources directly in the definition of the data-based predictor without increasing the dimension of the optimization problem to be solved. This choice translates into considering a weighted average of the available data, with weights computed a priori by heuristically evaluating the similarity between the controlled plant (i.e., the nominal system) and the others based on the available data. We then provide an analytical analysis of the impact of our design choices on the statistical features of the data used for control. Through a numerical example, we show that using data from multiple sources might be a viable option to enhance tracking performance with noisy data at the price of an increased sensitivity to regularization penalties (often used in data-driven control to cope with noisy data~\cite{berberich2020data,Dorfler2023gap,Coulson2022}). 

\paragraph*{Outline} The paper is organized as follows. We introduce the considered setting and the goal of this work in Section~\ref{sec:setting}, while Section~\ref{sec:preliminaries} provides some preliminaries on DeePC. Then, Section~\ref{sec:federated_DeePC} presents the proposed federated DeePC approach, discussing the impact of our choices on the statistics of the data used for prediction. The impact of using the federated approach is then discussed in Section~\ref{sec:examples}. The paper ends with some concluding remarks and directions for future works.  

\paragraph*{Notation \& useful definitions} We denote the set of natural numbers (including zero) and real numbers as $\mathbb{N}$ ($\mathbb{N}_{0}$) and $\mathbb{R}$, respectively. Meanwhile, $\mathbb{R}^{n}$ and $\mathbb{R}^{n \times m}$ denote the space of real, $n$-dimensional column vector and the set of real, $n\times m$ dimension matrices, respectively. We indicate identity matrices as $I$ and zero vectors as $0$, without explicitly indicating their dimensions. For a matrix $A \in \mathbb{R}^{n\times m}$, $\|A\|_{2}$ denotes the induced 2-norm of matrix $A$, while $\|A\|_{2}^{2}=A^{\top}A$, with $A^{\top} \in \mathbb{R}^{m\times n}$ is the transpose of $A$. Given two integers $a<b\leq n$, $[A]_{a:b}$ represents the set of rows of $A$ from the $a$-th to the $b$-th. Given a vector $x \in \mathbb{R}^{n_x}$ and a matrix $Q\in \mathbb{R}^{n_x \times n_x}$, $\|x\|_{Q}^{2}=x^{\top}Qx$. We denote $w_k \in \mathbb{R}^{n_w}$ as the instance of a time-series at time $k \in \mathbb{N}_{0}$, while we compactly denote the column vector stacking the elements of the sequence $\{w_0,w_1,\ldots,w_{T-1}\}$ as $w_{[0,T-1]}$, namely
\begin{equation*}
    w_{[0,T-1]}=\begin{bmatrix} w_0^\top & w_1^\top &\cdots & w_{T-1}^\top \end{bmatrix}^\top,
\end{equation*}
and indicate the associated Hankel matrix with $L\leq T$ rows as
\begin{equation}\label{eq:hankel}
    \mathcal{H}_{L}(w_{[0,T-1]})=\begin{bmatrix}
        w_0 & w_1 & \cdots & w_{T-L}\\
        w_1 & w_2 & \cdots & w_{T-L+1}\\
        \vdots & \vdots & \ddots & \vdots\\
        w_{L-1} & w_{L} & \cdots & w_{T-1}
    \end{bmatrix}.
\end{equation}
We then say that $w_{[0,T-1]}$ is persistently exciting of $L$ is the following definition holds.
\begin{definition}[Persistency of excitation]
    A finite-length signal $w_{[0,T-1]}$ is persistently exciting of order $L$ if $\mathcal{H}_{L}(w_{[0,T-1]})$ defined in \eqref{eq:hankel} is full row rank, i.e., $
        \mathrm{rank}(\mathcal{H}_{L}(w_{[0,T-1]}))=n_wL$.
    \hfill $\square$
\end{definition}
Given a linear, time-invariant system described by the following state-space model
\begin{equation*}
\begin{aligned}
x_{t+1}&=Ax_{t}+Bu_{t},\\
y_{t}&=Cx_{t}+Du_{t},
\end{aligned}
\end{equation*}
with $x_{t}\in \mathbb{R}^{n_x}$, $u_{t} \in \mathbb{R}^{n_u}$ and $y_{t} \in \mathbb{R}^{n_y}$, we define its extended observability matrix and the Toeplitz formed by its Markov parameters as follows:
\begin{equation}\label{eq:useful_matrices}
  \Gamma\!=\!\!\begin{bmatrix}
      C\\
      CA\vspace{-.1cm}\\
      \vdots\\
      CA^{T-1}
\end{bmatrix}\!,~\mathcal{T}\!\!=\!\!\begin{bmatrix}
D & 0  & \cdots & 0\\
CB & D  & \cdots & 0\vspace{-.1cm}\\
\vdots & \vdots & \ddots & \vdots \\
CA^{T\!-2}B & CA^{T\!-3}B & \cdots & D
  \end{bmatrix}\!.  
\end{equation}
Given a random vector $v \in \mathbb{R}^{n_v}$ we indicate its expected value as $\mathbb{E}[v] \in \mathbb{R}^{n_v}$.
\section{Setting \& goal}\label{sec:setting}
Consider a \emph{controllable}, linear, time-invariant (LTI) system $\mathcal{S}_{\mathrm{o}}$, with dynamics described by the following relationship
\begin{equation}\label{eq:dynamics}
\bar{y}_{t}=C_{\mathrm{o}}A_{\mathrm{o}}^{t}x_0+\sum_{\tau=0}^{t-1}C_{\mathrm{o}}A_{\mathrm{o}}^{\tau}B_{\mathrm{o}}u_{t-\tau-1}+D_{\mathrm{o}}u_{t},
\end{equation} 
where $\bar{y}_{t} \in\mathbb{R}^{n_y}$ is the (noiseless) output of the system at time $t \in \mathbb{N}_0$, $u_{t} \in \mathbb{R}^{n_u}$ is the associated input, $x_{0} \in \mathbb{R}^{n_x}$ is the initial state of the systems, while $(A_{\mathrm{o}},B_{\mathrm{o}},C_{\mathrm{o}},D_{\mathrm{o}})$ are the matrices characterizing a (minimal) state-space model for $\mathcal{S}_{\mathrm{o}}$ and, hence, its dynamic evolution. 

Let us assume that the latter matrices are \emph{unknown},  but we have access to a dataset $\mathcal{D}=\{\mathcal{D}_{i}\}_{i=0}^{M-1}$. Let such dataset consist of $M$ subsets $\mathcal{D}_{i}$ of input/output data, i.e., $\mathcal{D}_{i}=\{u_{[0,T-1]}^{d,(i)},y_{[0,T-1]}^{d,(i)}\}_{i=0}^{M-1}$, collected from both the \textquotedblleft nominal\textquotedblright \ system as well as other $M-1 \geq 1$ systems that are similar to the nominal one according to the following definition.
\begin{definition}[Similar systems]\label{def:similarity}
    Let $\Gamma_{1} \in \mathbb{R}^{n_yT\times n_x}$ and $\mathcal{T}_{1} \in \mathbb{R}^{n_yT\times n_uT}$ be the extended observability matrix and the Toeplitz matrix of the Markov parameters of an system $\mathcal{S}_{1}$ (see \eqref{eq:useful_matrices}). Consider another LTI system $\mathcal{S}_2$, sharing the same initial condition $x_{0}$ and input sequence $u_{[0,T-1]}$. Under these conditions, the LTI system $\mathcal{S}_2$ is said to be \emph{similar} to $\mathcal{S}_{1}$ if their noiseless outputs satisfy \begin{align}\label{eq:similarity_def}
        \nonumber \bar{y}_{[0,T-1]}^{(2)}&=(\Gamma_{1}+\Delta\Gamma_{21})x_{0}+\!(\mathcal{T}_{1}+\!\Delta\mathcal{T}_{21})u_{[0,T-1]}\\
        &=\bar{y}_{[0,T-1]}^{(1)}+\Delta y_{[0,T-1]}^{(2,1)},
    \end{align}
    where \textcolor{black}{$\Delta \Gamma_{21} = \Gamma_{2} - \Gamma_{1}$, $\Delta \mathcal{T}_{21} = \mathcal{T}_{2} - \mathcal{T}_{1}$, and}
    \begin{equation}
        \Delta y_{[0,T-1]}^{(2,1)}=\Delta\Gamma_{21}x_{0}+\Delta\mathcal{T}_{21}u_{[0,T-1]},
    \end{equation}
   with $\|\Delta y_{[0,T-1]}^{(2,1)}\|_{2}\leq \varepsilon_{y}$, 
   and 
   $\varepsilon_{y} \geq 0$ being sufficiently small, 
   while the two systems 
   retain the same properties (namely, controllability, observability and stability). \hfill $\square$
\end{definition} 

Based on this definition, suppose that the available data are collected when all the considered data-generating systems start from the same initial condition while being fed with the same input sequence, i.e., $u_{[0,T-1]}^{d,(i)}=u_{[0,T-1]}^{d}$ for all $i=0,\ldots,M-1$. \textcolor{black}{Note that, the data associated with the nominal system are always indicated via $i=0$.} Moreover, let us assume that the outputs collected from them are corrupted by noise, i.e.,
\begin{equation}\label{eq:noisy_output}
    y_{[0,T-1]}^{d,(i)}=\bar{y}_{[0,T-1]}^{d,(i)}+e_{[0,T-1]}^{(i)}, 
\end{equation}
where $e_t^{(i)} \in \mathbb{R}^{n_y}$ is the realization of the \emph{zero-mean}, \emph{white} measurement noise with covariance $\sigma_{i}^2 I$ acting on the $i$-th system, for $i=0,\ldots,M-1$. 

Under the additional assumption that the input sequence used for data collection is persistently exciting of a sufficient order (specified in Section~\ref{sec:preliminaries}), our goal is to design a data-driven predictive controller for the nominal systems $\mathcal{S}_{\mathrm{o}}$ by leveraging all the information at our disposal, ultimately searching for a preliminary answer to the following question: \emph{Can data from similar systems be leveraged to improve the performance of a data-driven predictive controller?} 

\begin{remark}[On our definition of similarity]
According to Definition~\ref{def:similarity}, two systems are similar if the differences in their dynamic response are dictated by bounded (and \textquotedblleft sufficiently small\textquotedblright) parametric uncertainties. Future work will be devoted to extending this definition, as well as considering alternative definitions of (input/output) similarity (see, e.g.,~\cite{fazziDistanceProblemsBehavioral2023,wangLearningControlSimilarity2024}). \hfill $\square$  
\end{remark}

\section{Preliminaries}\label{sec:preliminaries}
Before introducing our \textquotedblleft federated\textquotedblright \ strategy for data-driven predictive control, we provide a set of preliminaries by drawing from the seminal works in \cite{coulson2019data,berberich2020data}. Both of them rely on Willems' fundamental lemma~\cite{BehavioralSystemsTheoryWillems1997} to replace the predictor used in Model Predictive Control (MPC)~\cite{rawlings2009MPC} with one relying on raw data only. To this end, they rest on the assumption that the system under control is indeed controllable, and that the designer has access only to an input/output trajectory of the system $\{u_{[0,T-1]}^{d},y_{[0,T-1]}^{d}\}$. Such a trajectory is obtained by feeding the controlled system with a persistently exciting input of an order greater than $T_{\mathrm{ini}}+N+n_{x}$, with $n_x$ being the order of the unknown system, $N$ the (pre-fixed) prediction horizon of the predictive control scheme, and $T_{\mathrm{ini}}$ dictating the amount of past inputs and outputs needed to define the initial condition of the predictive control scheme for all time instant $t \in \mathbb{N}_0$ as follows:
\begin{equation}\label{eq:initial_condition}
u_{\mathrm{ini},t}\!=\!\begin{bmatrix}
       u_{t-T_{\mathrm{ini}}} \!\!\!\!\!& \cdots &\!\!\!\! u_{t-1}
    \end{bmatrix},~y_{\mathrm{ini},t}\!=\!\begin{bmatrix}
       y_{t-T_{\mathrm{ini}}} \!\!\!\!\!& \cdots &\!\!\!\! y_{t-1}
    \end{bmatrix}.
\end{equation}
This further implies that the length of the input/output trajectory $T\geq (n_u+1)(T_{\mathrm{ini}}+N+n_{x})$.
The available data are used to construct the following Hankel matrices 
\begin{subequations} \label{eq:Up_Yp_Matrices}
\begin{align}
    &U_P\!=\![\mathcal{H}_L(u_{[0,T-1]}^d)]_{1:T_{\mathrm{ini}}}~U_F\!=\![\mathcal{H}_{L}(u_{[0,T-1]}^d)]_{T_{\mathrm{ini}}+1:L} \\
    &Y_P\!=\![\mathcal{H}_{L}(y_{[0,T-1]}^d)]_{1:T_{ini}}~Y_F\!=\![\mathcal{H}_{L}(y_{[0,T-1]}^d)]_{T_{\mathrm{ini}}+1:L} \end{align}
\end{subequations}
with $L=T-T_{\mathrm{ini}}+N$, that, according to~\cite{BehavioralSystemsTheoryWillems1997}, are used to construct the data-driven predictor as
\begin{equation}\label{eq:data_driven_predictor}
    \begin{bmatrix}
        U_{P}\\
        Y_{P}\\
        U_{F}\\
        Y_F
        \end{bmatrix}g_t=\begin{bmatrix}
        u_{\mathrm{ini},t}\\
        y_{\mathrm{ini},t}\\
        u_{f,t}\\
        y_{f,t}
    \end{bmatrix},~~~~ t\in \mathbb{N}_0.
\end{equation}
with $g_t \in \mathbb{R}^{T-T_{\mathrm{ini}}-N}$, $u_{f,t}=u_{[t,t+N-1]}$ and $y_{f,t}=y_{[t,t+N-1]}$. The latter allows for the definition of the data-enabled predictive control (DeePC) problem, namely
\begin{subequations}\label{eq:standard_DeePC}
\begin{align}
    & 
    \min_{\substack{u_{f,t},~y_{f,t}\\g_t}}~J(u_{f,t},y_{f,t})+\lambda_g\|g_t\|_2^2\\
    & ~~~~~ \mbox{s.t.}~~\begin{bmatrix}
        U_{P}\\
        Y_{P}\\
        U_{F}\\
        Y_F
        \end{bmatrix}g_t=\begin{bmatrix}
        u_{\mathrm{ini},t}\\
        y_{\mathrm{ini},t}\\
        u_{f,t}\\
        y_{f,t}
    \end{bmatrix},\label{eq:prediction_model}\\
    &\quad \qquad~~ u_{t+k} \in \mathcal{U},~~k \in [0,N-1],\\
    &\quad \qquad~~ y_{t+k} \in \mathcal{Y},~~k \in [0,N-1],
\end{align}    
with 
\begin{equation}\label{eq:cost_perf}
    J(u_{f,t},y_{f,t})=\!\!\sum_{k=0}^{N-1} \|y_{t+k}-y_{t+k}^{r}\|_{Q}^{2}+\|u_{t+k}-u_{t+k}^{r}\|_{R}^{2},
\end{equation}
\end{subequations}
which is solved at each time instant $t \in \mathbb{N}_0$ in a receding horizon fashion.

The cost of the DeePC problem depends on $y_{[t,t+N-1]}^{r} \in \mathbb{R}^{n_y}$ and $u_{[t,t+N-1]}^{r} \in \mathbb{R}^{n_u}$, i.e., the output and input references the designer aims the system to track in closed-loop over the prediction horizon, as well as $Q \in \mathbb{R}^{n_y \times n_y}$ and $R \in \mathbb{R}^{n_u \times n_u}$, that are user-defined positive definite matrices\textcolor{black}{, which respectively penalize tracking error and control effort in the predicted trajectory}. Differently from standard MPC, the cost further features a regularization term \footnote{Additional slack variables can be used to consider the noise on $y_{\mathrm{ini},t}$, $Y_P$ and $Y_F$, that we do not include in this work. Note that, the noise on $y_{\mathrm{ini},t}$ can be counteracted with a proper choice of $T_{\mathrm{ini}}$ (see~\cite[Remark 2]{BRESCHI2023110961}).} that, contingent on the user-defined penalty $\lambda_g \geq 0$, contrasts the effect of noise on the available data. Moreover, the problem depends on 
$\mathcal{U} \subseteq \mathbb{R}^{n_u}$ and $\mathcal{Y} \subseteq \mathbb{R}^{n_y}$, which characterize the input and output constraints the designer aims to enforce.
\section{Toward Federated DeePC}\label{sec:federated_DeePC}
Instead of relying only on the data coming from the system to be controlled as in \eqref{eq:standard_DeePC}, we now wish to leverage all the information available in our dataset $\mathcal{D}$ (see Section~\ref{sec:setting}). Although an option could be to construct a bigger Hankel matrix by concatenating the Hankels associated with the data in each of the available subsets $\mathcal{D}_{i}$, $i=0,\ldots,M-1$, this approach is likely to make the resulting DeePC scheme numerically intractable as the optimization variable $g_t$ (see \eqref{eq:data_driven_predictor}) has a dimension that scales linearly with the number of data.

As a preliminary step towards federated DeePC, we propose to replace the predictor in \eqref{eq:prediction_model} with a \textquotedblleft federated\textquotedblright \ one, defined as
\begin{equation}\label{eq:federated_predictor}
    \sum_{i=0}^{M-1}\alpha^{(i)\!\!}\begin{bmatrix}
        U_P\\
        Y_P^{(i)}\\
        U_{F}\\
        Y_F^{(i)}
    \end{bmatrix}g_t\!=\!\begin{bmatrix}
        U_P\\
        Y_P^{\mathrm{fed}}\\
        U_{F}\\
        Y_F^{\mathrm{fed}}
    \end{bmatrix}g_t\!=\!\begin{bmatrix}
        u_{\mathrm{ini},t}\\
        y_{\mathrm{ini},t}\\
        u_{f,t}\\
        y_{f,t}
    \end{bmatrix}\!,
\end{equation}
where $\alpha^{(i)} \in [0,1]$, for $i=0,\ldots,M-1$, are user-defined weights satisfying 
\begin{equation}\label{eq:sum1}
    \sum_{i=1}^{M-1}\alpha^{(i)}=1,
\end{equation}
and $Y_{P}^{(i)}$ and $Y_{F}^{(i)}$ are 


obtained by decomposing the Hankel matrices of the outputs available in our dataset. In this work, we heuristically define the weights characterizing \eqref{eq:federated_predictor} as 
\begin{subequations}
\label{eq:penalties_emp}
\begin{equation}\label{eq:euristic_choice}
    \alpha^{(i)}=\left[\sum_{j=0}^{M-1}\exp{\left(-\beta\Delta \mathcal{H}_{L}^{(j)}\right)}\right]^{-1}\exp\left(-\beta \Delta \mathcal{H}_{L}^{(i)}\right),
\end{equation}
where 
\begin{equation} \label{eq:2norm-difference}
    \Delta \mathcal{H}_{L}^{(i)}=\left\|\mathcal{H}_{L}\left(y_{[0,T-1]}^{d,(i)}\right)-\mathcal{H}_{L}\left(y_{[0,T-1]}^{d,(0)}\right)\right\|_{2}
\end{equation}
\end{subequations}
for $i=0,\ldots,M-1$, and $\beta \geq 0$ is a tunable parameter. According to this definition, we thus tend to prioritize (i.e., weight more) the data collected from the nominal system and those that are \textquotedblleft closer\textquotedblright \ (according to the induced 2-norm) to the data generated by $\mathcal{S}_{\mathrm{o}}$. At the same time, by varying the value of the tunable parameter $\beta$, we can recover the original DeePC algorithm (for $\beta \rightarrow \infty$) or consider Hankel matrices comprising the sampled mean of the available outputs (when $\beta=0$). 

\textcolor{black}{The role of these weights as defined in \eqref{eq:euristic_choice} is to create a weighted average of the available data, where 
datasets that are more ``similar" to the nominal one are weighted the most. We acknowledge that our choice of 
weighting 
is dependent on initial conditions and specific input sequences, which might limit its applicability. Moreover, we also acknowledge that the weighting in \eqref{eq:euristic_choice} is affected by noise on the observed outputs. Nonetheless, the impact of noise on the weights is likely to be negligible as long as differences between systems dominate the effect of noise on the 2-norm in \eqref{eq:2norm-difference}. Nonetheless, the proposed 
heuristic measure of similarity is a stepping stone 
to show potential in weighting the separate data-sets differently, and future research will work on refining and formalizing the weightings of these matrices even further.
}

\begin{algorithm}[!tb]
\caption{Federated DeePC}
\label{algo:1}
\begin{algorithmic}[1] 
    \REQUIRE $\mathcal{D}$, $\lambda_g$, $Q$, $R$, $(u_{t}^{r},y_t^r)$ for all $t \in \mathbb{N}_0$
    \STATE \textbf{Compute} $\{\alpha^{(i)}\}_{i=0}^{M-1}$ as in \eqref{eq:euristic_choice};
    \STATE \textbf{Construct} the federated outputs as in \eqref{eq:federated_outputs};
    \FOR{$t=0,1\ldots T_{sim}$}
    \STATE \textbf{Construct} $u_{\mathrm{ini},t}$ and $y_{\mathrm{ini},t}$ as in \eqref{eq:initial_condition};
    \STATE \textbf{Solve} \eqref{eq:standard_DeePC} with the federated predictor in \eqref{eq:federated_predictor};
    \STATE \textbf{Get} its optimal solution $g_t^{\star}$; 
    \STATE \textbf{Find} the optimal input $u_t^\star=[U_F]_{1:n_u}g_t^{\star}$;
    \STATE \textbf{Feed} $u_t^\star$ to $\mathcal{S}_{\mathrm{o}}$;
    \ENDFOR
\end{algorithmic}
\end{algorithm}

By relying on the definition of the federated predictor in \eqref{eq:federated_predictor}, with weights defined as in \eqref{eq:euristic_choice}, our proposal for a federated DeePC scheme is summarized in Algorithm~\ref{algo:1}. Note that, since the construction of the federated predictor takes place before closing the loop, federated DeePC shares the same computational complexity of the standard DeePC scheme.  

\subsection{The impact of exploiting similarities in DeePC}
We now analyze how the definition of the federated predictor in \eqref{eq:federated_predictor} impacts the statistical properties of the data we use in the DeePC scheme, allowing us to unveil the possible benefits and drawbacks of using data from multiple systems as proposed. To this end, according to Definition~\ref{def:similarity}, we recast the measured outputs as follows
\begin{equation}\label{eq:meas_output}
    y_{[0,T-1]}^{d,(i)}=\bar{y}_{[0,T-1]}^{d,(0)}+\Delta y_{[0,T-1]}^{d,(i,0)}+e_{[0,T-1]}^{(i)},
\end{equation}
where $\bar{y}_{[0,T-1]}^{d,(0)}$ is the noiseless output of the nominal system, while 
\begin{equation}\label{eq:discrepancy_bounds}
 \|\Delta y_{[0,T-1]}^{d,(i,0)}\|_{2}\leq \varepsilon_{y}^{(i)}  \leq 
 \varepsilon_{y}, 
\end{equation}
for all $i=0,\ldots,M-1$, with $\varepsilon_{y}<\infty$. Accordingly, the data comprised in $Y_{P}^{\mathrm{fed}}$ and $Y_{F}^{\mathrm{fed}}$ are ultimately given by:
\begin{align}\label{eq:federated_outputs}
    \nonumber &y_{[0,T-1]}^{d,\mathrm{fed}}=\sum_{i=0}^{M-1}\alpha^{(i)}y_{[0,T-1]}^{d,(i)}\\
    &~=\bar{y}_{[0,T-1]}^{d,(0)}\!+\!\!\sum_{i=1}^{M-1}\!\alpha^{(i)}\Delta y_{[0,T-1]}^{d,(i,0)}\!+\!\!\sum_{i=0}^{M-1}\!\alpha^{(i)} e_{[0,T-1]}^{(i)},
\end{align}
where the first term of the last equality is obtained thanks to \eqref{eq:sum1}. Accordingly, by relying on the fact that the measurement noises are all assumed to be zero mean, it is straightforward to prove that
\begin{equation}\label{eq:mean_fed}
    \mathbb{E}[y_{[0,T-1]}^{d,\mathrm{fed}}]=\bar{y}_{[0,T-1]}^{d,(0)}\!+\!\!\sum_{i=1}^{M-1}\!\alpha^{(i)}\Delta y_{[0,T-1]}^{d,(i,0)},
\end{equation}
by relying on the properties of the expected value and the fact that $\alpha^{(i)}$ is deterministic for all $i=0,\ldots,M-1$. This straightforward result allows us to formalize the following bound on the distance between the expected value of the federated data and the noiseless output of the nominal system.
\begin{proposition}
    The federated output sequence in~\eqref{eq:federated_outputs} satisfies
    \begin{equation}\label{eq:bound_mean}
        \left\|\mathbb{E}[y_{[0,T-1]}^{d,\mathrm{fed}}]-\bar{y}_{[0,T-1]}^{d,(0)}\right\|_{2}\leq \sum_{i=1}^{M-1}\alpha^{(i)}\varepsilon^{(i)},
    \end{equation}
    with $\varepsilon^{(i)}$ defined as in \eqref{eq:discrepancy_bounds}, for $i=0,\ldots,M-1$.
\end{proposition}
\begin{proof}
    By replacing the definition of $\mathbb{E}[y_{[0,T-1]}^{d,\mathrm{fed}}]$ into the left-hand-side of \eqref{eq:bound_mean} it easily follows that 
    \begin{align*}
        \left\|\mathbb{E}[y_{[0,T-1]}^{d,\mathrm{fed}}]-\bar{y}_{[0,T-1]}^{d,(0)}\right\|_2&=\left\|\sum_{i=1}^{M-1}\!\alpha^{(i)}\Delta y_{[0,T-1]}^{d,(i,0)}\right\|_2\\
        & \leq \sum_{i=1}^{M-1}\!\alpha^{(i)}\left\|\Delta y_{[0,T-1]}^{d,(i,0)}\right\|_2,
    \end{align*}
    where the first equality stems from the fact that $\Delta y_{[0,T-1]}^{d,(0,0)}=0$, while the second inequality is obtained by using the triangular inequality and the properties of the weights $\{\alpha^{(i)}\}_{i=1}^{M}$. By using \eqref{eq:discrepancy_bounds} then \eqref{eq:bound_mean} immediately follows, thus concluding the proof. \end{proof}

This result ultimately indicates that the Euclidean distance between the expected value of the federated output trajectory and the noiseless one of the nominal system is driven by the dissimilarity between the (noiseless) output trajectories as well as the user-defined weights $\{\alpha^{(i)}\}_{i=1}^{M}$. To contain such distance, the latter penalties should thus be chosen to prioritize data collected from systems that are more similar to the nominal one, while neglecting the others (as we do with the penalties defined in \eqref{eq:penalties_emp}). 

Note that the bound in \eqref{eq:bound_mean} can be extended to the actual federated output trajectory (and not its mean) for the specific case in which $\beta$ in \eqref{eq:penalties_emp} is zero, as formalized in the following corollary.
\begin{corollary}[An asymptotic result with $\beta=0$]
    Assume that the measurement noises acting on the $M$ available data-generating systems are mutually independent and identically distributed. Let the penalties used to construct the federated predictor in \eqref{eq:federated_predictor} be defined as in \eqref{eq:penalties_emp}, with $\beta=0$. Then, for $M\rightarrow \infty$, the following holds:
    \begin{equation}\label{eq:bound_fed_asymp}
        \left\|y_{[0,T-1]}^{d,\mathrm{fed}}-\bar{y}_{[0,T-1]}^{d,(0)}\right\|_{2}\leq \frac{1}{M}\sum_{i=1}^{M-1}\varepsilon^{(i)},
    \end{equation}
    with $\varepsilon^{(i)}$ defined as in \eqref{eq:discrepancy_bounds}, for $i=0,\ldots,M-1$.
\end{corollary}
\begin{proof}
    By \eqref{eq:penalties_emp}, when $\beta=0$ then the federated output in \eqref{eq:federated_outputs} becomes:
    \begin{equation*}
        y_{[0,T-1]}^{d,\mathrm{fed}}=\frac{1}{M}\sum_{i=0}^{M-1}y_{[0,T-1]}^{d,(i)}.
    \end{equation*} 
    Thanks to the assumption on the measurement noises, we can thus apply the law of large numbers, hence obtaining 
    \begin{equation*}
        y_{[0,T-1]}^{d,\mathrm{fed}}\underset{M\rightarrow \infty}{\longrightarrow} \bar{y}_{[0,T-1]}^{d,(0)}\!+\!\frac{1}{M}\sum_{i=1}^{M-1}\!\Delta y_{[0,T-1]}^{d,(i,0)}.
    \end{equation*}
    The result in \eqref{eq:bound_fed_asymp} then straightforwardly follows from the triangular in inequality and \eqref{eq:discrepancy_bounds}, concluding the proof. \hfill
\end{proof}
Hence, as the number of similar systems from which we can gather data increases, if these data are uniformly weighted, the only \textquotedblleft source of error\textquotedblright \ in reconstructing the noise-free trajectory of $\mathcal{S}_{\mathrm{o}}$ is given by the discrepancies between the latter and the (noiseless) outputs of the other systems.
\begin{remark}[Federated \emph{vs} averaged DeePC]
    Constructing the federated predictor with $\beta=0$ allows us to remove (in mean) the effect of measurement noise if we have enough (theoretically infinite) similar systems that we can draw data from. This result comes at the price of retaining an error induced by the system's dissimilarities with the nominal one. On the other hand, if one can repeat the \emph{same} experiment on the nominal system enough times, the effect of noise can again be counteracted (in mean) without retaining the dissimilarity-driven error. However, this experiment design choice entails that the nominal system cannot be used for purposes other than generating data, making this option less feasible from a practical perspective than collecting data from similar (yet not identical) systems. \hfill $\square$
\end{remark}

Let us now consider the dispersion of the federated output trajectory with respect to the nominal (noiseless) one, namely
\begin{equation}\label{eq:disperision_definition}
    \Omega\!=\!\mathbb{E}\!\left[\!\left(y_{[0,T-1]}^{d,\mathrm{fed}}\!-\bar{y}_{[0,T-1]}^{d,0}\right)\left(y_{[0,T-1]}^{d,\mathrm{fed}}\!-\bar{y}_{[0,T-1]}^{d,0}\right)^{\!\!\top}\right].
\end{equation}
By exploiting the definition of $y_{[0,T-1]}^{d,\mathrm{fed}}$ in \eqref{eq:federated_outputs} and the properties of the measurement noises (specifically, its independence) 
, it is straightforward to prove that \eqref{eq:disperision_definition} can be equivalently rewritten as
\begin{align}\label{eq:disperision_definition2}
     \nonumber \Omega=&\sum_{i=1}^{M-1}\!\!\left(\alpha^{(i)}\right)^{\!2}\Omega_{\Delta}^{(i)\!}+2\!\!\sum_{i=1}^{M-1}\sum_{\substack{j=1\\j\neq i}}^{M-1}\!\alpha^{(i)}\alpha^{(j)}\Omega_{\Delta}^{(i,j)}\\
     &+\sum_{i=0}^{M-1}\!\left(\alpha^{(i)}\right)^{\!2}\sigma_{i}^2 I,
\end{align}
with 
\begin{align}
    &\Omega_{\Delta}^{(i)}=\Delta y_{[0,T-1]}^{d,(i,0)}\left(\Delta y_{[0,T-1]}^{d,(i,0)}\right)^{\!\top},\label{eq:delta_i}\\
    &\Omega_{\Delta}^{(i,j)}=\Delta y_{[0,T-1]}^{d,(i,0)}\left(\Delta y_{[0,T-1]}^{d,(j,0)}\right)^{\!\top},\label{eq:delta_ij}
\end{align}
is the (deterministic) dissimilarity-induced dispersion, which is a $n_y \times n_y$ non-negative matrix. When looking at the dispersion of the federated output trajectories as well as the mean, the weights $\alpha^{(i)}$ should thus be chosen considering the quality of data as well as similarity
. Indeed, one should not only retain data from systems that are more similar to the nominal one, but also those that are less affected by noise. This result becomes even clearer when looking at the following proposition. 
\begin{proposition}
    Let $\rho_{\mathrm{max}}(\Omega_{\Delta}^{(i)})$ and $\rho_{\mathrm{max}}(\Omega_{\Delta}^{(i,j)})$  be the largest singular value of $\Omega_{\Delta}^{(i)}$ in \eqref{eq:delta_i} and $\Omega_{\Delta}^{(i,j)}$ in \eqref{eq:delta_ij}, respectively. Then, the dispersion $\Omega$ in \eqref{eq:disperision_definition} satisfies
    \begin{align}\label{eq:dispersion_bound}
        \nonumber \|\Omega\|_2\leq& \left(\alpha^{(0)}\right)^{\!2}\sigma_0^2+\!\!\sum_{i=1}^{M-1}\left(\alpha^{(i)}\right)^{\!2}\left[\rho_{\mathrm{max}}\left(\Omega_{\Delta}^{(i)}\right)\!+\sigma_i^2\right]\\
        & +\!\sum_{i=1}^{M-1}\sum_{\substack{j=1\\j\neq i}}^{M-1}\alpha^{(i)}\alpha^{(j)}\rho_{\mathrm{max}}\left(\Omega_{\Delta}^{(i,j)}\right).
    \end{align}    
\end{proposition}
\begin{proof}
   The proof straightforwardly follows from the definition of $\Omega$ in \eqref{eq:disperision_definition2}, the recursive application of the triangular inequality, as well as the properties of the induced 2-norm of a matrix. 
\end{proof}

The result in \eqref{eq:dispersion_bound} ultimately allows us to detect in which situation using the federated output could lead to a dispersion around $\bar{y}_{[0,T-1]}^{d,0}$ that is smaller than the one attained by using the trajectory from the nominal system only (i.e., $\|\Omega\|_{2}=\sigma_0^2$). Indeed, when $\alpha^{(0)} \in [0,1)$, and the following holds
\begin{align}\label{eq:bound_variance1}
    \nonumber &\sum_{i=1}^{M-1}\tilde{\alpha}^{(i)\!\!}\left[\rho_{\mathrm{max}}\left(\!\Omega_{\Delta}^{(i)}\!\right)\!+\!\sigma_i^2\right]\\
    &\qquad \!+2\!\sum_{i=1}^{M-1}\sum_{\substack{j=1\\j\neq i}}^{M-1}\tilde{\alpha}^{(i,j)}\rho_{\mathrm{max}}\!\left(\!\Omega_{\Delta}^{(i,j)}\!\right)\!<\!\sigma_0^2,
\end{align}
with 
\begin{equation*}
    \tilde{\alpha}^{(i)}=\frac{\left(\alpha^{(i)}\right)^2}{1-\left(\alpha^{(0)}\right)^2},~~~~\tilde{\alpha}^{(i,j)}=2\frac{\alpha^{(i)}\alpha^{(j)}}{1-\left(\alpha^{(0)}\right)^2},
\end{equation*}
then $\|\Omega\|<\sigma_0^2$. This condition ultimately implies that using the federated scheme can be advantageous when data are collected from systems that are actually similar, with their measurement noise eventually being less severe than that corrupting the nominal system. In the specific case in which the measurement noises are identically distributed, \eqref{eq:bound_variance1} translates into: 
\textcolor{black}{\small
\begin{align}\label{eq:bound_variance2}
    &\sum_{i=1}^{M-1}\hat{\alpha}^{(i)\!}\rho_{\mathrm{max}}\left(\Omega_{\Delta}^{(i)}\right)\!+\!\! \sum_{i=1}^{M-1} \!\sum_{\substack{j=1\\j \neq i}}^{M-1}\hat{\alpha}^{(i,j)\!}\rho_{\mathrm{max}}\left(\Omega_{\Delta}^{(i,j)}\right) \!\!<\!\sigma_0^2,
    \normalsize
\end{align} }
with
\begin{equation*}
    \hat{\alpha}^{(i)}=\frac{\left(\alpha^{(i)}\right)^{\!2}}{1-\sum_{i=0}^{M-1}\left(\alpha^{(i)}\right)^{\!2}},~~~~\hat{\alpha}^{(i,j)}=2\frac{\alpha^{(i)}\alpha^{(j)}}{1-\sum_{i=0}^{M-1}\left(\alpha^{(i)}\right)^{\!2}},
\end{equation*}
thus indicating the level of dissimilarity that guarantees a smaller dispersion $\Omega$ with respect to that achieved by using the trajectories from the nominal system only. Note that, if $\alpha^{(0)}=1$, then the federated scheme corresponds back to the standard DeePC.

\section{Numerical case study}\label{sec:examples}
We now evaluate the performance of the proposed formulation for federated DeePC\footnote{\textcolor{black}{The MatLab Source code for generating these results can be found in \texttt{www.github.com/PichiControlEngineering/F-DeePC}}} by considering the nominal system with dynamics described by the following matrices (see \eqref{eq:dynamics})
\begin{equation}\label{eq:nominal_example}
\begin{aligned}
    &A_{\mathrm{o}}=\begin{bmatrix}
        0.7326 & -0.0891\\
        0.1722 & 0.9909
    \end{bmatrix},~~B_{\mathrm{o}}=\begin{bmatrix}
        0.0609\\
        0.0064
    \end{bmatrix},\\
    &C_{\mathrm{o}}=\begin{bmatrix}
        0 & 1
    \end{bmatrix},~~D_{\mathrm{o}}=0.
\end{aligned}
\end{equation}
The system is excited with a zero-mean white noise input with variance $1$ for $T=50$ steps, starting from $x_0= {\left[
    1 \; 1 \right]}^{\top}$, to get the output trajectories used to construct the data-driven predictor in \eqref{eq:prediction_model}. The same conditions are used to generate the other subset composing the \textquotedblleft federated\textquotedblright \ dataset $\mathcal{D}$ by considering additional $M-1$ systems sharing the same state-space matrices as \eqref{eq:nominal_example} but the state transition matrix, with the latter given by
\begin{equation}\label{eq:uncertain_A}
    A^{(j)\!}\!=\!A_{\mathrm{0}}\!+\!0.05\left(2\frac{j\!-\!1}{M\!-\!1}-1\right)\Delta A, ~~ j\!=\!1,\ldots,M\!-1.
\end{equation}
with $\Delta A$ defined such that these systems can be regarded as similar to the nominal one according to Definition~\ref{def:similarity}. In line with the setup introduced in Section~\ref{sec:setting}, the output sequences collected from the systems are corrupted by zero-mean white noise with different variances chosen for the Signal-to-Noise Ratio (SNR) of all outputs to be $20$~dB. Closed-loop simulations are instead not affected by noise toward being able to analyze the impact of using the federated approach without it being masked by additional noise.

Within this setting, we design an (unconstrained) federated DeePC scheme presented in Section~\ref{sec:federated_DeePC}, as well as the nominal DeePC (i.e., solving \eqref{eq:standard_DeePC} with the data from the nominal plant) to regulate the nominal system to zero, by setting $Q=I$, $R=0.01$ and $N=T_{\mathrm{ini}}=3$. This last choice guarantees that $T_{\mathrm{ini}}$ is greater than the system lag, while not becoming larger than the chosen prediction horizon, as well as simultaneously containing computation times required to solve \eqref{eq:standard_DeePC}. Performance is assessed over $150$ datasets obtained via Monte Carlo simulations (with new realizations of inputs and noise used for the batch data generation), considering noiseless closed-loop simulation and fixing $\beta=0.1$ in \eqref{eq:euristic_choice}

\paragraph*{Performance assessment} 
\begin{figure}[tb]
    \centering
    \begin{tabular}{c}
    \subfigure[]{
    \includegraphics[width=0.8\linewidth, trim={0 0 0 0},clip]{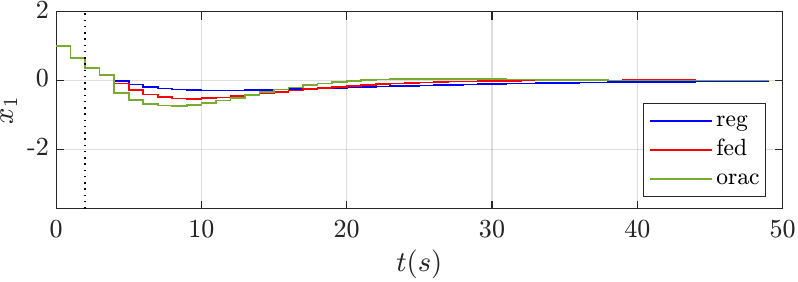}} \\
    \subfigure[]{
    \includegraphics[width=0.8\linewidth]{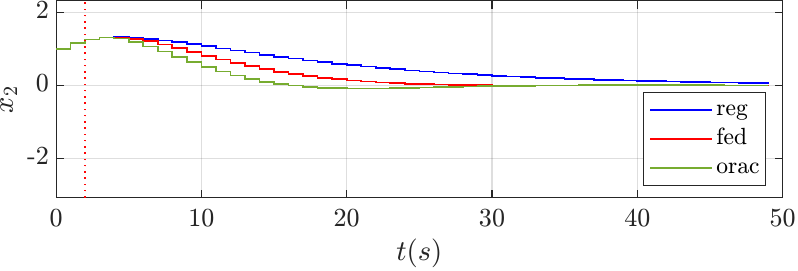}} \\
    \subfigure[]{
    \includegraphics[width=0.8\linewidth]{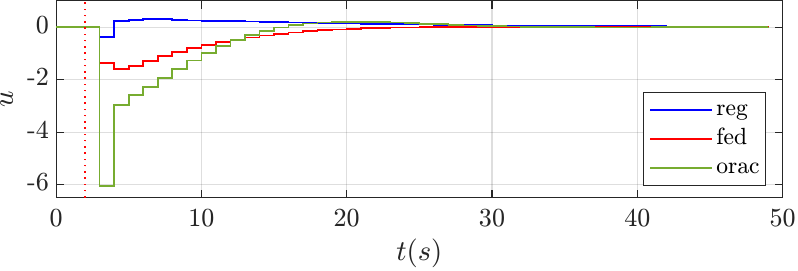}}
    \end{tabular}\vspace{-.2cm}
    \caption{An example trajectory highlighting performance of the system in \eqref{eq:nominal_example} with the 3 different controllers. The federated dataset consists of $M=55$ systems perturbed by $\Delta A_2$. For each controller, $\lambda_g$ is chosen such that it minimizes the mean ${RMSE}_y$}
\label{fig:example_trajectory}
\end{figure}
In evaluating the performance of the federated DeePC scheme presented in Section~\ref{sec:federated_DeePC}, we compare it against the nominal DeePC scheme (i.e., solving \eqref{eq:standard_DeePC} with the data from the nominal plant) and what we refer to as the \textquotedblleft oracle\textquotedblright, namely the DeePC controller designed by relying on noiseless data collected from \eqref{eq:nominal_example} in the same setting described above, setting $\lambda_g=0$ in \eqref{eq:standard_DeePC}. Note that the predictor used in the oracle DeePC thus corresponds to the actual prediction model of the LTI system (see \cite{coulson2019data}). In particular, from a quantitative standpoint, we focus on comparing the performance of the nominal and the federated DeePC scheme against the oracle by computing the following indicators:
\begin{equation}\label{eq:performance_indexes}
    \begin{aligned}
        & \mathrm{RMSE}_\zeta=\sqrt{\frac{1}{T_{\mathrm{sim}}}\!\!\sum_{t=0}^{T_{\mathrm{sim}-1}}(\zeta_t-\zeta_t^{\mathrm{o}})^2},\\
        & \mathrm{RMSE}_\zeta^{\mathrm{fed}}=\sqrt{\frac{1}{T_{\mathrm{sim}}}\!\!\sum_{t=0}^{T_{\mathrm{sim}-1}}(\zeta_t^{\mathrm{fed}}-\zeta_t^{\mathrm{o}})^2},
    \end{aligned}
\end{equation}
where $\zeta_t$, $\zeta_{t}^{\mathrm{fed}}$ and $\zeta^{\mathrm{o}}_t$ are placeholders for either the inputs or the outputs attained in closed-loop by using the standard DeePC scheme, the federated one, or the oracle, respectively, for $t=0,\ldots,T_{\mathrm{sim}}-1$, while $T_{\mathrm{sim}}=50$ is the considered number of simulation steps. 

\subsection{Performance with different perturbations}
We first consider two possible options for the uncertainty $\Delta A$ in \eqref{eq:uncertain_A}, namely
\begin{equation*}
    \Delta A_1=\begin{bmatrix}
        0 & 1\\
        -1 & 0
    \end{bmatrix},~~\Delta A_2=I,
\end{equation*}
evaluating the performance of the standard DeePC algorithm against the federated one for different values of $\lambda_g$ when $M=55$. An example trajectory for a control sequence obtained in this setting is shown in 
\figurename{~\ref{fig:example_trajectory}}
.

\begin{figure*}[tb!]
    \centering
    \begin{tabular}{ccc}
        \includegraphics[scale=.4]{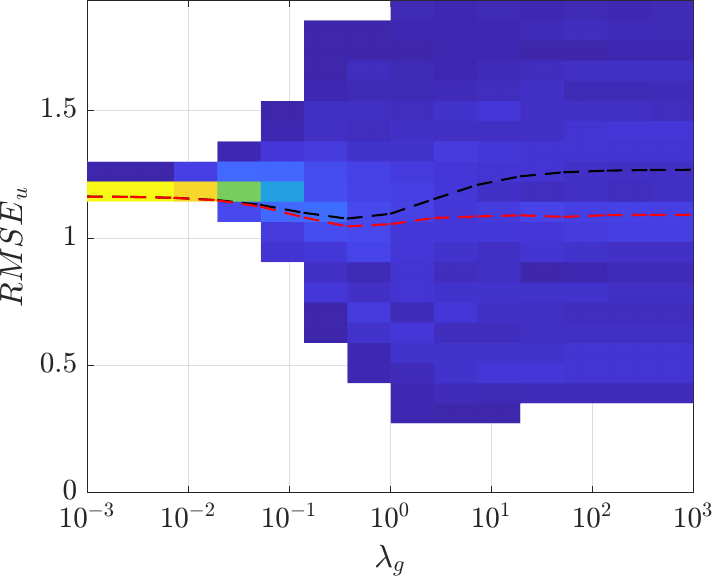}  & \includegraphics[trim={0 0 2.1cm 0},clip,scale=0.4]{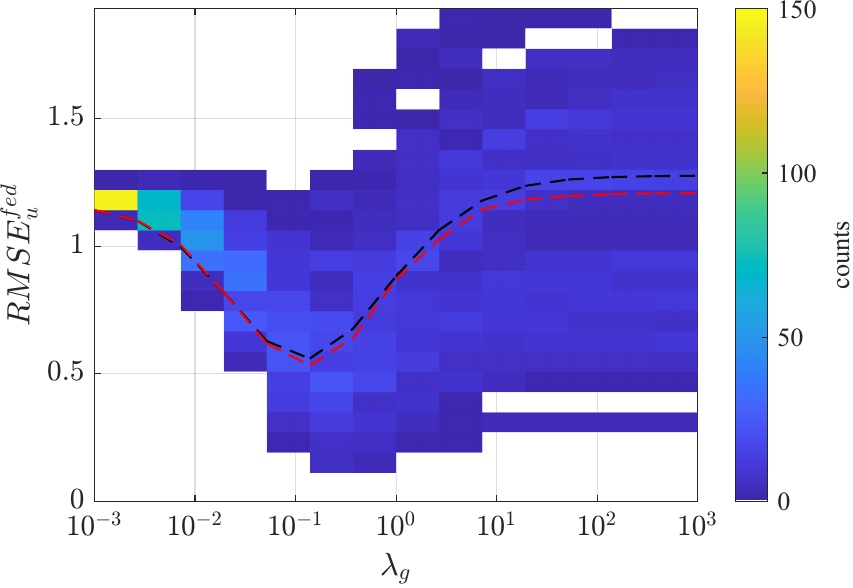} & \includegraphics[scale=.4]{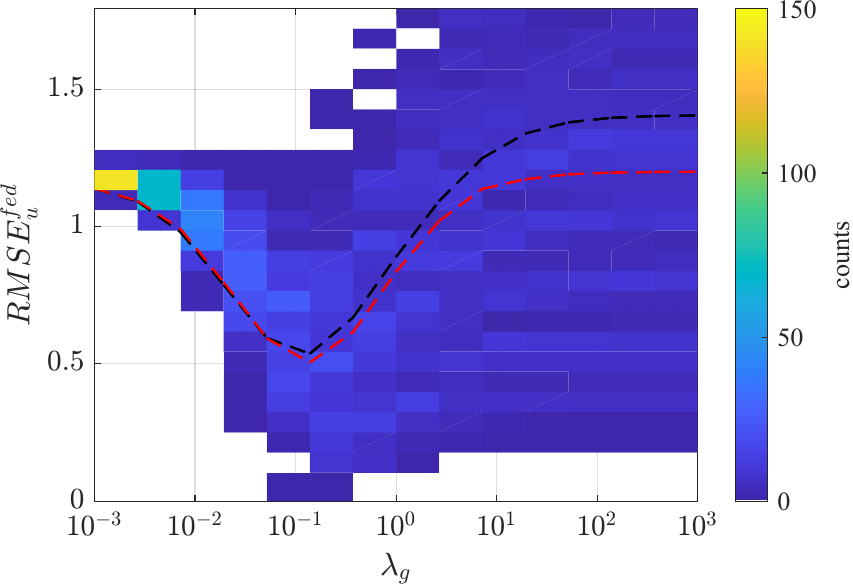}\\
        \subfigure[Standard DeePC]{\includegraphics[scale=.4]{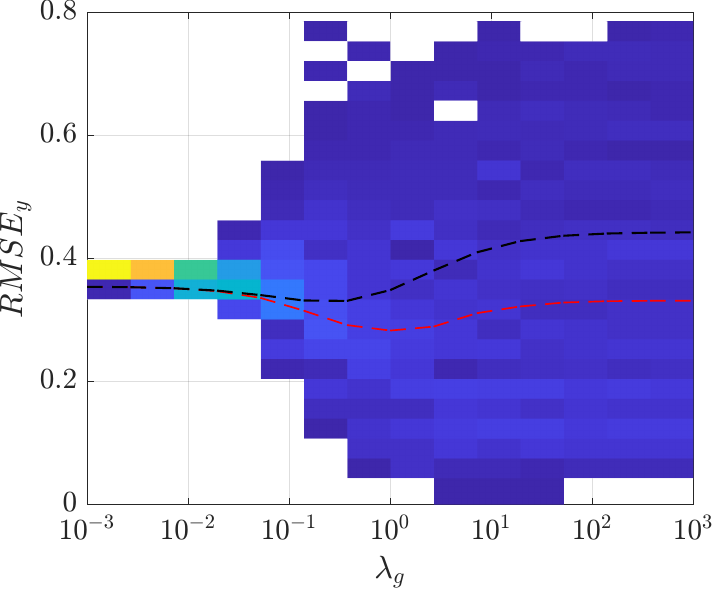}} & \subfigure[Federated DeePC with $\Delta A_1$]{\includegraphics[trim={0 0 2.1cm 0},clip,scale=.4]{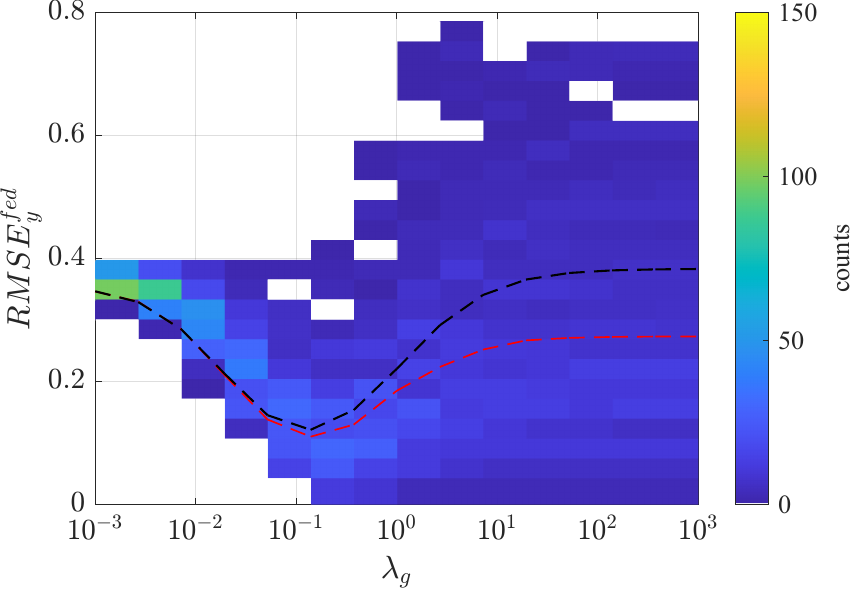}} & \subfigure[Federated DeePC with $\Delta A_2$]{\includegraphics[scale=.4]{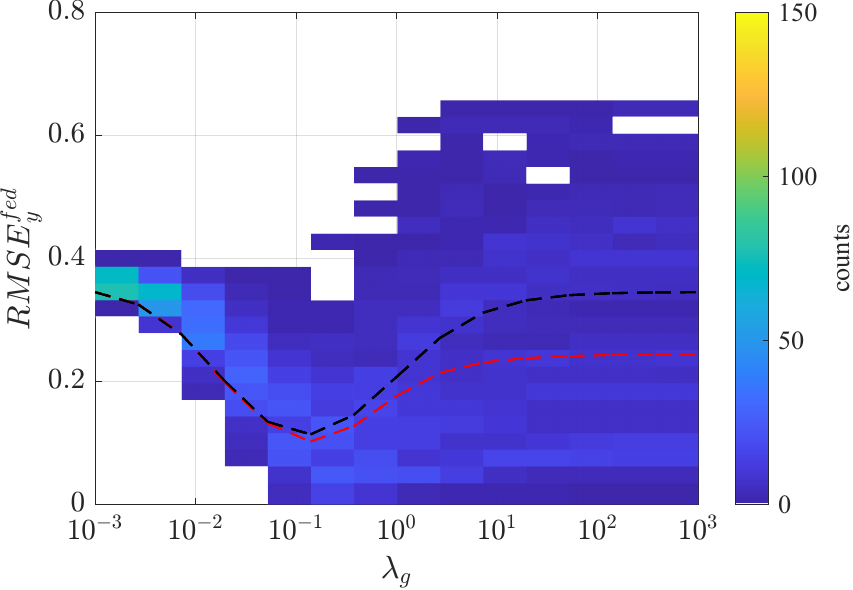}}
    \end{tabular}
    \vspace{-.2cm}\caption{Histograms of the values taken by the performance indicators in \eqref{eq:performance_indexes} for different values of $\lambda_g$ over $150$ Monte Carlo simulation, with mean (dashed black) and median (dashed red) of the achieved RMSEs.}
    \label{fig:rms_compare}
\end{figure*}

As shown in \figurename{~\ref{fig:rms_compare}}, the performance of standard DeePC and the two considered federated cases are comparable for the lower values of $\lambda_g$ tested. Nonetheless, by using the federated approach, one could better identify the value of $\lambda_g$ leading to the best performance in terms of RMSEs, which are visibly better than the ones attained with standard DeePC. When focusing on the case $\Delta A=\Delta A_1$ in \eqref{eq:uncertain_A} and looking at the median output RMSEs, the performance of the federated DeePC scheme tends to remain superior to those achieved with standard DeePC. For high regularization values, both controllers tend to generate zero inputs. Meanwhile, when looking at $\Delta A=\Delta A_2$, the median and mean output RMSEs tend to be more coherent, with federated DeePC performing better or comparably to standard DeePC in mean over the whole spectrum of tested $\lambda_g$, but leading to a slightly worse result than the standard DeePC in mean.

These results highlight that, at least in the considered case studies, using the federated DeePC scheme is advantageous both in terms of performance and in easing the detection of the optimal penalty $\lambda_g$. At the same time, this result implies that the federated approach is, in this case, more sensitive to the choice of $\lambda_g$, especially for higher values of such a hyperparameter when $\Delta A=\Delta A_1$. Therefore, the federated approach requires more careful tuning of $\lambda_g$, in this case avoiding excessive regularization and, hence, reduction in the magnitude of the optimization variable $g$. 
It is worth remarking that he \textquotedblleft federated\textquotedblright \ RMSEs is overall lower than the one attained with the standard DeePC scheme but eventually features more prominent outliers.

\subsection{The impact of increasing amount of accessible systems}
\begin{figure}[!tb]
    \centering
\includegraphics[scale=0.65]{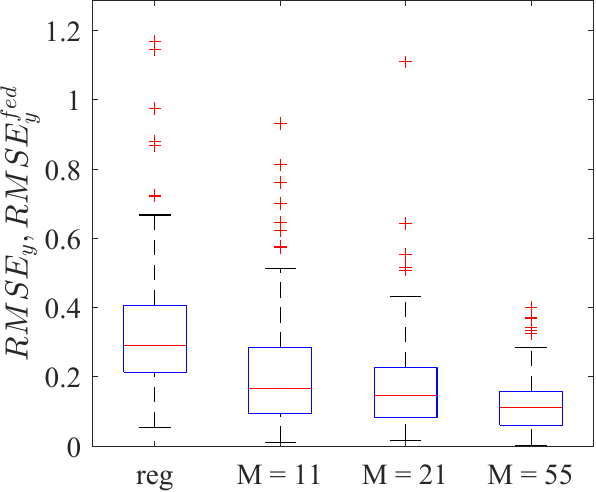}\vspace{-0.2cm}
    \caption{Box-plot of the RMSE$_y$ and RMSE$_y^{\mathrm{fed}}$ for the standard DeePC scheme (reg) and the federated one for increasing $M$ over 150 Monte Carlo simulations.}
    \label{fig:boxPlot}
\end{figure}
\begin{figure}[!tb]
    \centering
    \includegraphics[scale=.65]{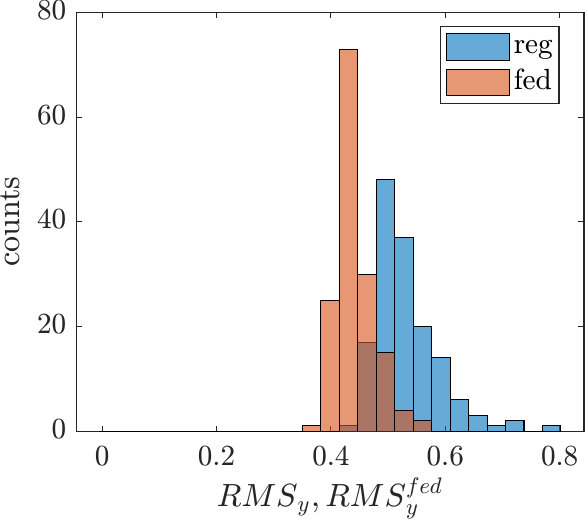}
    \caption{Distribution of the output RMS for standard DeePC (reg) with the (optimal) $\lambda_g=0.37$ and federated DeePC with (optimal) $\lambda_g=0.14$ and $M=55$ over 150 Monte Carlo simulations.}
    \label{fig:absoluteRMS}
\end{figure}
By focusing on $\Delta A=\Delta A_1$ and considering the values of $\lambda_g$ leading to the minimum RMSEs only, we now analyze the impact of increasing the number of similar systems available to construct the federated predictor in \eqref{eq:federated_predictor}. As is clear from \figurename{~\ref{fig:boxPlot}}, in the considered setting, increasing the number of similar systems used to construct the predictor allows the achievement of closed-loop outputs that are progressively closer to those attained with the oracle control scheme. Moreover, increasing $M$ also reduces outliers in the output RMSE, showing the possible benefits of using the federated approach in terms of dispersion of performance around the median error with respect to the oracle. A similar conclusion can be drawn when specifically looking at the absolute root mean square (RMS) tracking error, i.e., 
\begin{equation}\label{eq:performance_indexes2}
    \begin{aligned}
        & \mathrm{RMS}_y=\sqrt{\frac{1}{T_{\mathrm{sim}}}\!\!\sum_{t=0}^{T_{\mathrm{sim}-1}}(y_t-y_t^{r})^2},\\
        & \mathrm{RMS}_y^{\mathrm{fed}}=\sqrt{\frac{1}{T_{\mathrm{sim}}}\!\!\sum_{t=0}^{T_{\mathrm{sim}-1}}(y_t^{\mathrm{fed}}-y_t^{r})^2},
    \end{aligned}
\end{equation}
for the case $M=55$ (see \figurename{~\ref{fig:absoluteRMS}}). Indeed, the distribution of the output RMSEs achieved with the federated scheme is centered around values lower than those attained with the standard DeePC approach, showing that the federated scheme allows us not only to attain closer oracle tracking, but also improve overall performance. We wish to stress that these results do not imply that one gets to the oracle solution when $M\rightarrow \infty$, which is only possible if all the systems are exactly equal and the associated weights in the federated predictor \eqref{eq:federated_predictor} are uniformly chosen (see~\eqref{eq:bound_fed_asymp}).

\section{Conclusions}\label{sec:conclusions}
In this paper, we have presented a scheme for the design of a federated data-driven predictive controller. We have proposed a preliminary approach to integrating data collected from similar systems into the standard DeePC framework, analyzing the impact that our choice has on the data matrices used to construct the data-driven predictor. Based on our integration choices, we show that these additional data might be advantageous in reducing the dispersion around the noiseless output of the nominal plant to be controlled while introducing a structural bias due to the plants' dissimilarities, which can be compensated for by a proper choice of the weights characterizing our data integration strategy. Our numerical results show the possible advantages and drawbacks of exploiting the proposed federated scheme.

Future work will be devoted to refining our notion of similarity and, consequently, improving our federated strategy, as well as exploring alternative approaches to leverage data from multiple (similar) systems for data-driven predictive control, e.g., the similarity measure introduced in \cite{padoan2022}.

\bibliography{toward_federatedDeePC.bib}

\end{document}